\documentclass[11pt]{amsart}
\usepackage{amsfonts,amscd,amsthm,amsgen,amsmath,amssymb, hyperref}
\usepackage[all]{xy}
\usepackage[vcentermath]{youngtab}
\usepackage[letterpaper,hmargin=1.55in,vmargin=1.2in]{geometry}
\usepackage{graphicx, xcolor}
\usepackage{bm}	
\usepackage{amsmath}
\usepackage{amsthm}	
\usepackage{amsfonts}	
\usepackage{amssymb}
\usepackage{mathrsfs}
\usepackage{latexsym}
\usepackage{enumerate}
\usepackage{dsfont}
\usepackage{amssymb}

\theoremstyle{plain}
\newtheorem{theorem}{Theorem}[section]

\newtheorem{lemma}[theorem]{Lemma}

\theoremstyle{definition}

\theoremstyle{remark}

\numberwithin{equation}{section}
\numberwithin{figure}{section}

\begin{document}

\title{Non-associative magnetic translations: A QFT construction}

\author{Jouko Mickelsson}
\address{Department of Mathematics and Statistics, University of Helsinki}

\email{jouko@kth.se}

\maketitle
\begin{abstract} 

The non-associativity of translations in a quantum system with magnetic field
background has received renewed interest in association with topologically trivial
gerbes over $\mathbb{R}^n.$ The non-associativity is described by a 3-cocycle
of the group $\mathbb{R}^n$ with values in the unit circle $S^1.$ 
The gerbes over a space $M$ are topologically classified by the Dixmier-Douady
class which is an element of $\mathrm{H}^3(M,\mathbb{Z}).$ However, there is a finer
description in terms of local differential forms of degrees $d=0,1,2,3$ and
the case of the magnetic translations for $n=3$ the 2-form part is the magnetic field
$B$ with non zero divergence.  In this paper we study a quantum field theoretic 
construction in terms of $n$-component fermions on a real line or a unit circle.
The non associativity arises when trying to lift the translation group action
on the 1-particle system to the second quantized system. 


\end{abstract}

\section{Introduction}

The motivation for the present short note is to understand the recent paper by
Bunk, M\"uller and Szabo \cite{BMS} in terms of quantization of Dirac operators
on a real line or on the circle coupled to an abelian vector potential with gauge
group $\mathbb{R}^n$ or the torus $T^n=\mathbb{R}^n/\mathbb{Z}^n.$
The central topic in \cite{BMS} is a 3-cocycle on $\mathbb{R}^n$ arising from composing certain
functors coming from translations acting on differential data of a topologically
trivial gerbe on $\mathbb{R}^n.$ The non-associativity in the case of a magnetic field with
sources in the case $n=3$ was suggested already long ago in \cite{Jac}. An interpretation
of the 3-cocycle in terms of representations of canonical anticommutator algebras
was then proposed in \cite{Car}.

In this paper we interpret the magnetic translations as (non periodic) gauge transformations on real line acting on fermions with $n$ complex components. They actually define true operators on the level
of 1-particle Dirac operators. However, they cannot be lifted to unitary
operators in the fermionic Fock space; if they could, there would be no 3-cocycle
since the composition of linear operators is associative. Nevertheless, these
gauge transformations define functors acting on certain categories of representations
of canonical anticommutation relations. The composition of functors respects
the group law in $\mathbb{R}^n$ only modulo the action of automorphisms in the
Fock space; these automorphisms come from a projective representation of an abelian
gauge group. 

Let $G$ be a  simply connected Lie group.  Let $H$ be the space of square integrable functions
on the real line taking values in the complex vector space $\mathbb{C}^n$ with a
unitary $G$ action through a representation $\rho$ of $G.$  The group $LG$ of smooth $G$ valued functions $f$ on $\mathbb{R}$
such that $f(t)$ is constant $g$ outside of a compact set  acts unitarily on $H.$
Using the stereographic projection from the unit circle to the real axis we can
actually identify $LG$ as a subgroup of the  smooth loop group on the unit circle.

Let $C_g$ be the category of smooth  paths $f,$ parametrized by a closed interval of the real axis, in $G$ starting from the unit element $e$
and with the end point $g$  with vanishing derivatives at
the points $e,g.$  Morphisms in the category $C_g$ are smooth homotopies of paths with
fixed end points. Next we choose a representation of the canonical anticommutation relations in a fermionic Fock space $\mathcal{F}_f$ 
with a Fock vacuum defined by a polarization $H= H_{+}(f) \oplus H_-(f)$ of $H.$ 
Starting from the polarization $H=H_+\oplus H_-$ defined by the Fourier decomposition
to non negative and negative Fourier modes we set $H_+(f) = f\cdot H_+$ and $H_-(f)$
its orthogonal complement.  [Alternatively, for the purposes of the present note, we could consider polarizations of the one dimensional Dirac operator $D_f = i\frac{d}{dx}+ + i f^{-1}df$].  

The fixed polarization $H_+ \oplus H_-$ defines a representation of the canonical
commutation relations generated by the elements
$a^*(v), a(v)$ for $v\in H$ with nonzero anticommutation relations 
$$a^*(u) a(v) + a(v) a^*(u) = 2 <u,v>_H$$
and a Fock vacuum $\psi$ with $a^*(u) \psi =0 = a(v)\psi$ for $u\in H_-$ and $v\in H_+.$ 

We have a functor from the category of paths $C_g$ to the category of CAR algebra
representations $\mathcal{F}(g)$ sending $f$ to $\mathcal{F}_f.$ The CAR representations
in the category $\mathcal{F}(g)$ are all equivalent. Two paths $f,f' \in C_g$ are related by
a point-wise multiplication by an element $h$ of the loop group $LG.$ An element of the loop
group is represented as an unitary operator $T(h)$ in the Fock space and
$T(h) a^*(v) T(h)^{-1} = a^*(hv).$ However, the operator $T(h)$ is fixed only up to a
phase due to the fact that the loop group is projectively represented, through a
central extension $\widehat{LG}.$ For this reason the functor $F$ is projective in the sense that morphisms
in the category $C_g$ go over to morphisms (unitary equivalences) in $\mathcal{F}(g)$
respecting the composition only up to a phase.

Each element $g_1\in G$ defines a functor $F_{g_1}: C_g \to C_{gg_1}$ as follows.
Fix a path $f_1$ joining $e$ to $g_1.$ Take any $f\in C_g,$ parametrized by an
interval $[a,b].$ Parametrize $f_1$ by an interval $[b,c].$ Then joining the
two paths gives a new path $f*f_1$ as follows: First travel $f$ until the end point
$g= f(b). $ Then continue with $t\mapsto f(b)f_1(t)$ for $b \leq t \leq c$ ending
at $gg_1.$  The functor $F_{g_1}$ from $C_g$ to $C_{gg_1}$ defines naturally also
a functor from $\mathcal{F}(g)$ to $\mathcal{F}(gg_1).$ Namely, the path $f_1$ joining
$e$ to $g_1$ defines an automorphism of the CAR algebra by $a(v) \mapsto a(gf_1v)$ 
taking a representation in the category $\mathcal{F}(g)$ to a representation in
the category $\mathcal{F}(gg_1).$ 

Next fix a pair $g_1, g_2 \in G.$ Choose as above a pair of paths $f_1,f_2$
parametrized by the intervals $[b,c]$ and $[c,d]$ correspondingly. On the other hand,
we have a path $f_{12}$ joining $e$ to $g_1 g_2.$ Finally, we have a loop $\ell(g_1, g_2)$ by composing $f_1 * f_2 * f_{12}^{-1};$ recall that functions on the real line
constant outside of a compact set can be identified as elements of the loop group.
The last factor involves the point-wise 
inverse of the function travelled in the opposite direction.  This loop construction
is similar but different from the construction in \cite{Mi09} where a functorial
approach to group 3-cocycles was discussed.  

The functor $F(g)$ can be represented as an operator by a point-wise multiplication
in the 1-particle Hilbert space $H.$ However, it does not define an operator in
a Fock space. The reason is that the off-diagonal blocks of the 1-particle operator
are not Hilbert-Schmidt with respect to energy polarization due to the non periodicity
of the path; this is seen by a simple
Fourier analysis using the polarization defined by $D= i\frac{d}{dx}.$

It follows directly from the definition that we have the 2-cocycle property
\begin{equation}
 \ell(g_1, g_2) \ell(g_1g_2,g_3) = \,\,^{g_1}\ell(g_2,  g_3)  \ell(g_1,g_2g_3).
\label{loopcocycle} \end{equation}
where $^g\ell$ denotes the left translate of the loop $\ell$ by $g.$ A remark about the
parametrizations:  The first loop on the left connects the points $e, g_1, g_1g_2, e$
at the parameter points $0, a_1, a_1 + a_2, 0$ and the second loop on the left
connects $e, g_1g_2, g_1g_2 g_3, e$ at parameter values $0, a_1 +a_2, a_1 +a_2 +a_3, 0$
so the point-wise product connects $e, g_1, g_1g_2, g_1 g_2 g_3, e$ at parameter values
$0, a_1, a_1 +a_2, a_1 +a_2 + a_3, 0$; the paths between $e, g_1g_2$ in each factor
cancelling since they are inverse of each other. The reader can check that the same result is obtained for the product on the right in \eqref{loopcocycle}.

Using the construction in \cite{Mi87} for any loop $\ell\in LG$ we can fix an element
in the standard central extension $\widehat{LG}$ of $LG$ by $S^1$ by choosing 
an extension $\tilde{\ell}$ to the unit disk; on the boundary $\tilde{\ell}$ is
equal  to $\ell.$ As a circle bundle, the central extension consists of equivalence
classes of pairs $(\tilde{\ell},\lambda)$ with $\lambda\in S^1$ with the equivalence
relation
$$(\tilde{\ell},\lambda) \sim (\tilde{\ell}', \lambda')$$
with $\lambda' = \lambda e^{2\pi i\int_V \Omega}$
where $V$ is a volume in $G$ with boundary obtained by glueing the surfaces $\tilde{\ell},\tilde{\ell'}$
along the common boundary $\ell$ and $\Omega$ is a representative of a class in
$\mathrm{H}^3(G, \mathbb{Z}).$ 

The 2-cocycle property above fails for the lifts of the loop group elements
to the central extension $\widehat{LG}.$ A triple $g_1,g_2, g_3$ determines
through the choices of the loops $\ell$ in \eqref{loopcocycle} and their extensions
$\tilde{\ell}$ a tetraed with faces given by the four extensions $\tilde{\ell}(g_1,g_2),
\tilde{\ell}(g_1g_2, g_3), \tilde{\ell}(g_2, g_3), \tilde{\ell}(g_1, g_2g_3).$
This closed 2-surface $\Sigma$ in $G$ is then equivalent to the phase
$$c(g_1, g_2, g_3)=\exp{2\pi i \int_V \Omega}$$ 
where $V$ is the volume in $G$ with boundary $\Sigma.$ This is the 3-cocycle which
comes from the extensions of the loops in \eqref{loopcocycle} to the central extension,
\begin{equation}
 \tilde{\ell}(g_1, g_2) \tilde{\ell}(g_1g_2) = [f_1 \tilde{\ell}(g_1, g_2 g_3) {f_1}^{-1}] \tilde{\ell}(g_2,g_3)
 \times c(g_1, g_2, g_3) \label{3cocycle}
 \end{equation}
 Because of the different choices made in the construction of \eqref{3cocycle} the
 cocycle $c$ is smooth only in an open neighborhood of the unit element
 in $G.$
 
 \section{The case of $G= \mathbb{R}^n$}
 
 The group $G=\mathbb{R}^n$ has an unitary representation in the Hilbert space
 $H$ of square integrable functions on $\mathbb{R}$ with values in $\mathbb{C}^n$ 
 through multiplication $z_k \mapsto e^{i x_k} z_k$ for $k=1,2, \dots ,n.$
 This defines also an action of the loop group $LG$ in $H$ through point-wise
 multiplication by the phase $e^{i x_k(t)}.$ 
 
 \begin{lemma} The group of continuous piecewise smooth loops satisfies the
 Hilbert-Schmidt condition on off-diagonal blocks for the energy polarization
 $H=H_+\oplus H_-.$ \end{lemma}
 
 \begin{proof} We set $n=1;$ the multicomponent case is proven in a similar way.
 The Fourier components of the multiplication operator $\ell(t)$ can be estimated
 by integration by parts:  restricting to any interval $[a,b]$ where $\ell$ is smooth
 we get for momenta $p,q$ of opposite sign
 $$ <p|\ell|q> = \int_a^b  e^{i(p-q)t} \ell(t) dt
 = \int_a^b \frac{1}{i(p-q)} e^{i(p-q)t} \ell'(t) dt + \frac{1}{i(p-q)}e^{i(p-q)t} \ell(t) \arrowvert_a^b.$$
 The valuations at the end points cancel when summing over all intervals for
 a periodic continuous function $\ell.$    In the first term on the right we can repeat
 the integration by parts. Now since the derivative $\ell'$ might be discontinuous
 at the end points of the intervals the insertion terms do not cancel. However,
 they involve the same factor $(p-q)^{-2}$  as in the integration term involving
 $\ell''(t).$ The Hilbert-Schmidt condition follows taking the square and observing
 that
 $$\int_{p  > \gamma, q < -\gamma}  \frac{1}{(p-q)^4} dp dq < \infty$$
 for any positive $\gamma,$ and likewise for $p < -\gamma, q > \gamma.$
 \end{proof}
 
For any $x\in \mathbb{R}^n$ choose the path $f_x$ as the straight line from the origin
to the point $x.$ Then proceeding as the general case above for any pair of vectors
$x,y\in \mathbb{R}^n$ we have the closed loop $\ell(x,y)$ as the triangle with vertices at $0, x, x+y.$ According to the Lemma this piece-wise smooth loop is represented
as a unitary operator in the Fock space defined by the energy polarization 
of the free Dirac operator $i\frac{d}{dt}$ on the real line. 

Next fix a a closed 3-form on $\mathbb{R}^n$ by $\Omega= \sum_{i,j,k} a_{ijk} dx_i \wedge
dx_j \wedge dx_k$ where $a$ is any antisymmetric tensor. This closed form is exact,
$\Omega = dB$ with $B= \sum_{i,j,k} a_{ijk} x_i dx_j \wedge dx_k.$ The forms $\Omega, B $
define a topologically trivial gerbe over $\mathbb{R}^n.$ 

For a pair $x,y$ of vectors the loop $\ell(x,y)$ is the boundary of a triangle
$\tilde \ell(x,y)$ and for a triple $x,y,z\in \mathbb{R}^n$ we have a tetraed
$V(x,y,z)$ with faces consisting of the triangles $\tilde{\ell}(x,y), \tilde{\ell}_x(y,z), \tilde{\ell}(x+y, z), \tilde{\ell}(x+y, z)$ where $\ell_x$ denotes the triangle $\ell$
translated by the vector $x.$ Thus the vertices of the tetraed $V(x,y,z)$ are located 
at the points $0, x, x+y, x+y+z.$ We observe
$$\int_{V(x,y,z)}  \Omega = \sum a_{ijk} x_i y_j z_k.$$ 
In particular, when $n=3$ and $a_{ijk} = \epsilon_{ijk}$  the value of the integral
is equal to the volume of the tetraed $V(x,y,z).$ As before, the corresponding 3-cocycle
is
$$c(x,y,z) = e^{2\pi i \sum_{ijk} a_{ijk} x_i y_j z_k}.\label{3cocycle} $$
Although this cocycle for (nonzero $a$) is nontrivial as a group cocycle it is
however trivial as a transformation groupoid cocycle: The group $\mathbb{R}^n$ 
acts on itself by translations and $c= \delta b$ for the the 2-cochain
$b(u; x,y) = c(u,x,y)$
with
$$(\delta b)(x,y,z) = b(u; x,y)^{-1}  b(u; x+y, z)^{-1}  b(u; x, y+z) b(u+x; y,z)$$

The cocycle $c$ is equal to the identity if all the vectors belong to the subgroup
$\mathbb{Z}^n \subset \mathbb{R}^n.$ In that case all the functors corresponding to the edges of the tetraed are actually loops in $T^n = \mathbb{R}^n/\mathbb{Z}^n$ 
and are represented by unitary operators in the Fock space. 

Following the rules of the canonical quantization of bounded operators in the 1-particle Hilbert space satisfying the Hilbert-Schmidt condition on the off-diagonal blocks
with respect to the energy polarization the Lie algebra is represented projectively 
in the fermionic Fock space. The projective action is characterized by the 2-cocycle \cite{Lu} 
$$c_2(f,g) = \frac12 \text{TR} \,f [\epsilon, g].$$
The trace is computed in the 1-particle Hilbert space. In particular, when the Lie
algebra consist of multiplication operators by smooth functions (on a circle or
on the real line, constant outside a compact set)  we have
$$\frac12 \text{TR} \, f[\epsilon, g] = \frac{1}{2\pi i} \int \text{tr} f dg$$
 where the second trace is evaluated in the representation $\rho$ of $G$ in $\mathbb{C}^n.$
 
In the present setting the loops take values in $\mathbb{R}^n$ and $\rho(x)z_k =
e^{ix_k} z_k$ for $k=1,2,\dots, n.$ 
Each component defines a circle value function $e^{i f(t)}$ acting as a multiplication operator
in the 1-particle Hilbert space $H=L_2(\mathbb{R}, \mathbb{C}^n).$ The cocycle $c_2$ is nontrivial on the abelian
loop group. However, in the case of a family of Dirac operators $D_A = i \frac{d}{dt}
+ A$ coupled to an abelian vector potential $A$ (with values in $\mathbb{R}^n$) the cocycle becomes trivial:
we have $c_2 = \delta b_1$ where 
$$b_1(A; X) = \frac{1}{4\pi i} \int \sum_k A_k X_k dt$$
and the loop algebra element $X$ acts on $A$ through the gauge transformation $A \mapsto A + dX.$ 
For this reason the bundle of Fock spaces parametrized by the vector potentials
becomes equivariant with respect to the gauge action and can be pushed forward
to a bundle over the flat moduli space $\mathbb{R}^n= \mathcal{A}/\mathcal{G}$ of gauge potentials;
here $\mathcal{G}$ is the group of periodic functions $f: \mathbb{R} \to \mathbb{R}^n$
(that is, $f$ is a constant outside of a compact set)
acting on potentials as 
$A\mapsto A+  df.$ 

\section{The case of a torus}

If we replace the gauge group $\mathbb{R}^n$ by the torus $T^n$ the situation
becomes different. All the maps $f:\mathbb{R}  \to \mathbb{R}^n$ which are periodic
modulo $\mathbb{Z}$ (that is, the asymptotic values of $f$ on the right in $\mathbb{R}$ are related to the values on the left by a shift in $\mathbb{Z}^n$) satisfy the Hilbert-Schmidt condition on off-diagonal blocks
with respect to the energy polarization; again, a function $f$ defines a multiplication
operator in the one-particle space through $z_k \mapsto e^{2\pi i  f_k}z_k.$ These functions $f$ van be viewed as loops $S^1 \to T^n.$
Now the group of gauge transformations $\mathcal{G}$ factorizes as a product of the
group $\mathcal{G}_0$ contractible
maps to $T^n$ (represented as loops on $\mathbb{R}^n$) and a group $\mathbb{Z}^n$
of maps  of
the form $f(t) = 0$ for $t\leq 0,$ $f(t) = t v$ for $0\leq t\leq 1$ with $v \in \mathbb{Z}^n$ and $f(t) =v$ for $t\geq 1.$ 

The moduli space of gauge potentials $\mathcal{A}/\mathcal{G}$ is now
the torus $T^n;$ we have $\mathcal{A}/\mathcal{G}_0 = \mathbb{R}^n$ and the second
factor in $\mathcal{G}$ is isomorphic to the subgroup $\mathbb{Z}^n \subset \mathbb{R}^n.$ In the case of $\mathbb{R}^n$ there was no restriction on the normalization of
the 3-cocycle [as a group cocycle or as a 3-form on $\mathbb{R}^n$] but in the case
of the torus the 3-cocycle must satisfy an integrality constraint in order that
the gerbe over $T^n$ is well-defined.

As explained in \cite{HaMi}, \cite{Mi17} (see also \cite{MiWa} Section 7) the 1-particle  Dirac hamiltonians can
be twisted in such a way that their K-theory class over the moduli space $T^n$ 
is nontrivial: the Chern character has a nonzero component $\omega_3$ in $\mathrm{H}^3(T^n, \mathbb{Z}).$ The basis in $\mathrm{H}^3(T^n,\mathbb{Z})$ is given by the
3-forms $\omega_3 = \sum a_{ijk} dx_i \wedge dx_j\wedge dx_k$ where the $a's$ form
a basis of totally antisymmetric tensors of rank 3 with integral coefficients,
The pull-back with respect to the projection $\mathbb{R}^n \to T^n$ is the form
$\sum a_{ijk} x_i dx_j\wedge dx_k.$ The quantum field theoretic construction
of a gerbe over the torus from a non zero class $[\omega_3]$ is recalled in the
Appendix. 

The 3-form part $\omega_3$ of the Chern character is the Dixmier-Douady class of
the projective vector bundle over $T^n$ obtained by canonical quantization of the
family of 1-particle Dirac operators. The pull-back of this bundle over $\mathbb{R}^n$
comes by projectivization of a vector bundle (the bundle of fermionic Fock spaces).
The group $\mathbb{Z}^n$ acts through an abelian extension of the Fock spaces.
The extension is defined by the 2-cocycle
$$ c_2(u; x, y)= e^{2\pi i \sum a_{ijk} u_i x_j y_k}$$
where $x,y\in \mathbb{Z}^n$ and $u\in\mathbb{R}^n$ and $\mathbb{Z}^n$ acts on
the functions of the vector $u$ as translations.

The 3-cocycle \eqref{3cocycle} is identically $=1$ when the arguments are in $\mathbb{Z}^n
\subset \mathbb{R}^n$ in conformity with the (projective) action of $\mathbb{Z}^n$
on the Fock spaces.

{\bf Remark} The cocycle $c_2$ is also a group cocycle even in the case of constant
coefficients (no group action on $u$) but since the coboundary operator is
different  the cohomology with variable coefficients is different from the cohomology
with constant coefficients.

\section{Appendix} 
In this appendix we briefly recall the quantum field theoretic construction of a
gerbe over the torus using a twisted family of CAR algebra representations,
\cite{HaMi}, \cite{Mi17}.

A hermitean complex line bundle $L$ over the torus $T^n$ is characterized by a class $\omega$ in
$\mathrm{H}^2(T^n, \mathbb{Z}).$ Parametrizing the circles in the torus by the interval
$[0,1]$  the 2-cohomology is spanned  by antisymmetric bilinear forms on $\mathbb{R}^n$ such that $\omega(x,y) \in \mathbb{Z}$ for $x,y\in
\mathbb{Z}^n.$ The pull-back of $L$ over $\mathbb{R}^n$ is trivial and
the sections of that line bundle are complex valued functions $\psi$ such that
$$\psi(x + z) = \psi(x) e^{2\pi i \omega(x,z)}$$
for $z\in \mathbb{Z}^n.$

Next we construct a family of fermionic Fock spaces parametrized by vectors in
$\mathbb{R}^n.$ For each $k=1,2, \dots, n$ and $u,v\in H$ let $a_k(v), a^*_k(u)$
be generators of a CAR algebra with nonzero anticommutators
$$a^*_k(u) a_k(v) + a_k(v) a^*_k(u) = 2 <u,v>.$$
The generators for different lower indices are assumed to  commute. It will be
convenient to compactify the real line to the unit circle so we can take $H=
L_2(S^1, \mathbb{C}^n)$ and we can work with the orthonormal basis of Fourier modes
in each of the $n$ directions.

We twist the CAR algebra by the line bundle $L.$ This means that the families of creation
and annihilation operators are sections of the tensor products of $L$ or its dual and the CAR algebra. 
The sections are $\mathbb{Z}^n$ equivariant functions on $\mathbb{R}^n,$ 
that is, for  $x\in\mathbb{R}^n$  and  for $z\in\mathbb{Z}^n$
$$ a^*_k(u, x+z) = a^*_k(u,x) e^{2\pi i \omega(x,z)}, \,\, a_k(u,x+z) = a_k(u,x) e^{-2\pi i \omega(x,z)}.$$
The right-hand-side of the canonical anticommutation relations, when evaluated at a point
$x\in\mathbb{R}^n,$ is multiplied by the pairing of sections of $L,L^*$ involved in the
construction of $a^*_k(u,x) = a^*_k(u)\otimes \psi(x)$ and of $a_k(v)\otimes \xi.$

The  Fock vacuum is again annihilated by $a^*(u,x)$ and $a(v,x)$ for $u\in H_-$ and
$v\in H_+.$ (One could generalize this construction by allowing the modes for 
different lower index $k$ be twisted by different line bundles.) 

Thus the states with net particle number $N$ in the Fock space are twisted by 
the $N$:th tensor power of $L.$  

The group $\mathbb{Z}^n$ acts as automorphisms of the twisted  CAR algebra by
$$g(p) a^*(u,x) g(p)^{-1}= a^*(p\cdot u, x+p) = e^{2\pi i \omega(x,p)} a^*(p\cdot u,x)$$
where $p$ acts on a function $u(\zeta)$ by multiplication by a phase, $u_k
\mapsto e^{2\pi i \zeta p_k} u_k,$
that is, the Fourier modes  are shifted by $p$ units.
Likewise, for the annihilation operators
$$g(p) a(u,x) g(p)^{-1} = e^{-2\pi i \omega(x,p)} a((-p)\cdot u, x).$$
The action of $g(p)$ in the Fock spaces parametrized by $x$ is now completely defined by
fixing the action on the vacuum vector $\psi.$ This is easiest done thinking the vectors
as elements in the semi-infinite cohomology (in physics terms, the 'Dirac sea').
For $n=1$ the vacuum is symbolically the semi-infinite  product
$$ \psi = a^*_0 a^*_{-1} a^*_{-2} \cdots$$
where the lower index refers to the Fourier modes in $L_2(S^1).$ For general $n$ the
vacuum is defined in a similar way inserting the non-negative Fourier modes for each
of the $n$ components.  The CAR generators are labelled by a double index $(k,j)$
with $k\in\mathbb{Z}$ and $j=1,2,\dots n. $ The action of $g(p)$ on the vacuum is now defined as a shift
operator: The index $k$ of the element $a^*_{k,j}$  is shifted by the integer  $p_j$ for $j=1,2, \dots, n,$ $k\mapsto k+ p_j.$ 
 
Because of the phase shifts when the CAR algebra generators are conjugated by $g(p)$
the product $g(p) g(q)$ is not equal to $g(p+q)$ but they differ by a $x$ dependent phase,
$$g(p)g(q) = C(x; p,q) g(p+q)= e^{2\pi i N \omega(x,p)} g(p+q)$$
where $N$ is the particle number of the state $g(q)\psi,$ that is, $N= \sum_{j=1}^{j=n} q_j.$

{\bf Remark} The projective vector bundles over $T^n$ are classified by elements of
$\mathrm{H}^3(T^n, \mathbb{Z}.$ Representatives of these elements can be written as
de Rham forms $\Omega = \sum a_{ijk} dx_ \wedge dx_j \wedge dx_k$ where the coefficients
$a_{ijk}$ are integers. The pull-back of $\Omega$ with respect to the projection
$\pi: \mathbb{R}^n \to T^n$ is $\pi^*\Omega = d\theta= d\sum a_{ijk} x_i dx_j \wedge dx_k.$
Evaluating $\theta$ for tangent vectors $u,v$ in the integral lattice $\mathbb{Z}^n
\subset \mathbb{R}^n$ and exponentiating gives the 2-group cocycle
$$C'(x; u,v) = e^{2\pi i \sum a_{ijk} x_i u_j v_k}$$
where the group $\mathbb{Z}^n$ acts on the vector $x$ by $x\mapsto x+u.$
According to the discussion in \cite{MiWa}, Section 7.1, there is a 1-1 correspondence
between the group cohomology $\mathrm{H}_{grp}^2(\mathbb{Z}^n, A)$ and the de Rham
cohomology $\mathrm{H}^3(T^n, \mathbb{Z})$ where $A$ is the $\mathbb{Z}^n$ module
of (smooth) functions $T^n\to S^1.$ The above map $\{c_{ijk}\} \to C'$
realizes this isomorphism.

{\bf Example}: When $n=3$ the cocycle $C$ is equivalent to $C'$ for the choice
$a_{ijk} =\alpha \epsilon_{ijk}$ where 
where $\epsilon$
is totally antisymmetric tensor with $\epsilon_{123}=1$ and $\alpha = 2(\omega_{12}
+\omega_{23} + \omega_{31})$ where $\omega= \omega_{12} dx_1\wedge dx_2 + \omega_{31}
dx_3\wedge dx_1 + \omega_{23}dx_2\wedge dx_3.$
This is seen by projecting the exponent in $C$ to its totally antisymmetric
component.

 \end{document}